\documentclass[runningheads]{llncs}

\usepackage{graphicx}
\usepackage[misc,geometry]{ifsym}
\usepackage{amsmath}
\usepackage{amssymb}
\usepackage{graphicx}
\usepackage{yhmath}
\usepackage{eqparbox}
\usepackage{hyperref}
\usepackage{enumerate}
\usepackage[shortlabels]{enumitem}
\usepackage{tikz}
\usepackage{todonotes}
\usepackage{comment}
\usepackage[T2A]{fontenc}
\usepackage[utf8]{inputenc}

\def\A{\mathcal{A}}
\def\B{\mathcal{B}}
\def\E{\mathcal{E}}

\def\G{\mathcal{G}}
\def\LL{\mathcal{L}}
\def\N{\mathbb{N}}
\def\R{\mathcal{R}}
\def\Z{\mathbb{Z}}

\def\ww{\textbf{w}}
\newcommand{\bwt}[2]{{\rm bwt}_{#1} \left( #2 \right)}
\newcommand{\ebwt}[2]{{\rm ebwt}_{#1} \left( #2 \right)}

\newcommand*{\eqmathbox}[2][M]{\eqmakebox[#1]{$\displaystyle#2$}}

\date{}
\title{Clustering of Return Words in Languages of Interval Exchanges} 

\author{Francesco Dolce\inst{1}
	\and
	Christian B. Hughes\inst{1}
	}

\institute{Czech Technical University in Prague (Czech Republic) \\
	\email{\{dolcefra, hughechr\}@fit.cvut.cz}}

\begin{document}

	\maketitle
	
	\begin{abstract}
		A word over an ordered alphabet is said to be \emph{clustering} if identical letters appear adjacently in its Burrows-Wheeler transform.
		Such words are strictly related to (discrete) interval exchange transformations.
		We use an extended version of the well-known Rauzy induction to show that every return word in the language generated by a regular interval exchange transformation is clustering, partially answering a question of Lapointe (2021).
	\end{abstract}
	
	\keywords{Interval exchange transformations \and Burrows-Wheeler transform \and Clustering words \and Return words}

	\section{Introduction}
	\label{sec:intro}
	
	Interval exchange transformations (IETs), first introduced by Oseledec~\cite{Oseledec66} in 1966, are defined by first partitioning an interval into subintervals, then translating each subinterval by a fixed permutation.
	They form an important class of dynamical systems that are studied from different perspectives: symbolic dynamics, combinatorics on words, ergodic theory, and others.
	A rich body of work has since explored various structural and combinatorial properties of these transformations.
	One can code IETs in a natural way to obtain sequences of linear complexity, including Sturmian sequences, which have been widely studied (see, e.g.,\cite{bifixcodesiets,FerencziZamboni08,KanelBelovChernyatev10}).
	
	The Burrows-Wheeler transform, introduced in~\cite{BurrowsWheeler94}, is a transformation used in data compression that first rearranges the letters of a word by lexicographically sorting all of its conjugates, then reads in this order the last letters of these conjugates.
	Clustering words are words whose Burrows-Wheeler transform consists of adjacent occurrences of identical letters.
	A link between clustering words and IETs has been developed in recent years (e.g.,~\cite{FerencziHubertZamboni23,lapointe2021perfectly}).
	In particular, each clustering word can be associated with a discrete interval exchange  transformation (see~\cite{FerencziZamboni13}).
	
	Return words to $w$ in a language are words that when preceded by $w$ are still in the language and end with $w$ as well (see precise definition later).
	In a 2021 paper~\cite{lapointe2021perfectly}, Lapointe asked whether return words of a symmetric IET are themselves perfectly clustering.
	That is, if such return words cluster in a way that corresponds to the symmetric permutation.
	
	In this paper, spurred by Lapointe's 2021 question, we show that all return words of interval exchange transformations satisfying the Keane condition~\cite{Keane75}, i.e, regular IETs, are clustering.
	Our result leverages and extends previous combinatorial and dynamical insights, particularly from a work by the first author and Perrin on a two-sided version of Rauzy induction on regular IETs~\cite{branching}.
	The main result of this contribution is the following.
	
	\begin{theorem}
		\label{thm:main}
		Return words in a language generated by a regular interval exchange transformation are clustering words.
	\end{theorem}
	
	Our approach to this theorem relies on Rauzy induction, a dynamical tool introduced in its one-sided version by Rauzy~\cite{Rauzy79} and subsequently extended in various ways.
	We consider a family of morphisms that, under certain assumptions, preserves clustering at every step of the induction.
	This result, together with the fact that one can obtain the cylinders of a regular IET through Rauzy induction, allows us to prove the theorem.
	
	We conclude this contribution by extending the link between IETs (resp. DIETs) and clustering words to clustering multisets of words instead.
	In order to do so, we introduce the notion of alsinicity, a generalization of the more well-known concept of dendricity.

	\section{Preliminaries}
	\label{sec:preliminaries}
	
	For all undefined terms, we refer the reader to~\cite{Lothaire2}.
	
	\paragraph{\bf Words.}
	
	An \emph{ordered alphabet} $\A = \{ a_1 < a_2 < \ldots < a_d \}$ is a (finite) set of symbols called \emph{letters} together with an order of its elements.
	The set of \emph{finite words} $\A^*$ over $\A$ is the free monoid with neutral element the \emph{empty word} $\varepsilon$.
	The product of two words $u,v \in \A^*$ is given by their composition $uv$.
	We denote by $\A^+$ the free semigroup over $\A$, e.g., $\A^+ = \A^* \setminus \{ \varepsilon \}$.
	The order on $\A$ is naturally extended to $\A^*$ by the lexicographic order.
	For a given word $w = w_0 w_1 \cdots w_{n-1}$, where each $w_i \in \A$, we denote by $|w|$ its length $n$, and by $|w|_u$ the number of times $u$ appears a factor of $w$.
	The Parikh vector of a word $w \in \A^*$ is the vector $\Psi_{\A}(w) \in \N^d$ defined as $(\Psi_{\A}(w))_a = |w|_a$.
	A word $w \in \A^*$ is \emph{pangrammatic} if $|w|_a > 0$ for every $a \in \A$.
	Unless stated otherwise we will always consider words pangrammatic over their alphabet.
	These vectors can be generalized in a natural way to multisets of words over the same ordered alphabet.
	
	A word $w$ is said to be \emph{primitive} if it is not the integer power of another word, i.e., if $w = u^k$ implies $k=1$.
	Two words $w, w'$ are \emph{conjugate} if $w = uv$ and $w'=vu$ for some $u,v \in \A^*$.
	If a word $w$ is primitive, then it has exactly $|w|$ distinct conjugates.
	A \emph{Lyndon word} is a primitive word that is minimal for the lexicographic order among its conjugates.
	
	A (right) \emph{infinite word} over $\A$ is a sequence $\ww = w_0 w_1 w_2 \cdots$, with $w_i \in \A$ for all $i$.
	An infinite word $\ww$ is \emph{eventually periodic} if $\ww = u v^\omega = u v v v \cdots$.
	An infinite word that is not eventually periodic is called \emph{aperiodic}.

	\paragraph{\bf Languages.}
	
	By \emph{language} we mean a factorial and bi-extendable set $\LL \subset \A^*$, i.e., such that, for every $w \in \A^*$, we have $v, aw, wb \in \LL$ for every factor $v$ of $u$ and for certain letters $a,b \in \A$.
	The language of an infinite word $\ww$ is the set $\LL(\ww)$ of all its factors, while the language of a finite word $w$ is defined as $\LL(w^\omega)$.
	A language $\LL$ is \emph{recurrent} if, for every $v \in \LL$, $vuv \in \LL$ for a certain word $u$.
	It is \emph{uniformly recurrent} if for every $v \in \LL$, there exists $N \in \N$ such that $v$ appears as a factor of every element of length $N$ in $\LL$.
	The set $\R_{\LL}(w)$ of (right) \emph{return words} to $w$ in $\LL \subset \A^*$
	is the set of words $u$ such that $wu \in \LL$ has exactly two occurrences of $w$ as factor: as a prefix and as a suffix.
	Formally $\R_{\LL}(w) = \{ u \in \A^* \, | \, wu \in (\LL \cap \A^* w) \setminus \A^+ w \A^+ \}$.
	When the language $\LL$ is clear, we will simply write $\R(w)$.

	\paragraph{\bf Morphisms.}
	
	A \emph{morphism} is a map $\varphi: \A^* \to \B^*$, with $\A, \B$ alphabets, such that $\varphi(\varepsilon) = \varepsilon$ and $\varphi(uv) = \varphi(u) \varphi(v)$ for every $u,v \in \A^*$.
	Given two distinct letters $a,b \in \A$, let us define the morphisms $\alpha_{a,b}, \tilde{\alpha}_{a,b} :\A^* \to \A^*$ as
	$$
	\alpha_{a,b} =
	\left\{
	\begin{array}{ll}
		a \mapsto a b & \\
		c \mapsto c, & \quad c \neq a
	\end{array}
	\right.
	\qquad
	\rm{and}
	\qquad
	\tilde{\alpha}_{a,b} =
	\left\{
	\begin{array}{ll}
		a \mapsto b a & \\
		c \mapsto c, & \quad c \neq a
	\end{array}
	\right..
	$$

	\paragraph{\bf Permutations.}
	
	When the letters of the alphabet are indexed $\{ a_1 < \ldots < a_d \}$, we identify $S_\A$ with $S_d$ and write $a_{\pi(d)}$ instead of $\pi(a_d)$.
	To describe a permutation $\pi \in S_{\A}$, we will use either the one-line notation or the cyclic one.
	For instance, the symmetric permutation defined by
	$a_{\pi(i)} = a_{d-i+1}$ for every $1 \le i \le d$
	will be denoted as either $(a_d, a_{d-1}, \ldots, a_1)$ or as the composition of the $2$-cycles $(a_1 a_d) (a_2 a_{d-2}) \cdots (a_{\frac{d}{2}} a_{\frac{d}{2}+1})$ if $d$ is even (if $d$ is odd, the last cycle is replaced by $(a_{\frac{d+1}{2}})$).
	A permutation is \emph{circular} if it has only one cycle.
	It is \emph{reducible} if $\{ a_1 < \cdots < a_k \}$ is invariant under $\pi$ for every $1 \le k <d$.

	\paragraph{\bf Burrows-Wheeler Transform.}
	
	The Burrows-Wheeler transform of a word $w \in \A^*$ is the word $\bwt{\A}{w}$ obtained by concatenating the last (not necessarily distinct) letters of the $|w|$ conjugates of $w$, sorted lexicographically on $\A$.
	
	\begin{example}
		\label{ex:sphynx}
		Consider the word $w = {\tt sphynx}$ on the standard ordered English alphabet $\E = \{ {\tt a} < {\tt b} < \ldots < {\tt z} \}$.
		Then $\bwt{\E}{w} = {\tt pysxnh}$.
	\end{example}
	
	The following results are well known (see, e.g.,\cite{ChrochemoreDesarmenienPerrin05,MantaciRestivoRosoneSciortino07,MantaciRestivoSciortino03}).
	
	\begin{proposition}
		\label{pro:bwt-conjugates}
		Two words $u,v$ over the same ordered alphabet $\A$ are conjugate if and only if $\bwt{\A}{u} = \bwt{\A}{v}$ 
	\end{proposition}
	
	\begin{proposition}
		\label{pro:bwt-primitive}
		Let $u \in \A^*$.
		A word $w$ is a conjugate of $u^p$ if and only if
		$\bwt{\A}{u} = b_1 \cdots b_{|u|}$ and
		$\bwt{\A}{w} = b_1^p \cdots b_{|u|}^p$, with $b_i \in \A$.
	\end{proposition}
	
	In~\cite{MantaciRestivoRosoneSciortino07} it is shown that, given an ordered alphabet $\A$, an extended version of the Burrows-Wheeler transform, denoted $\textrm{ebwt}$, gives a bijection between $\A^*$ and the multiset of Lyndon words over $\A$, where the conjugates, possibly of different length, are ordered using the $\omega$-order instead of the lexicographic one: $u \le_{\omega} v$ if $u^\omega \le v^\omega$.
	
	\begin{example}
		\label{ex:ebwt}
		Let $W$ be the multiset $\{ {\tt aac}, {\tt ab}, {\tt ab} \}$ of Lyndon words over $\A = \{ {\tt a} < {\tt b} < {\tt c} \}$.
		We have $\Psi_{\A}(W) = (4,2,1)$ and $\ebwt{\A}{W} = {\tt c bb aaaa}$.
	\end{example}
	
	Let $\pi$ be a permutation on $\A$.
	A word $w \in \A^{*}$ is said to be $\pi$\emph{-clustering} for $\A$ if
	$\bwt{\A}{w} = a^{k_1}_{a_{\pi(1)}} ... a^{k_r}_{a_{\pi(d)}}$, where $k_i = |w|_{a_{\pi(i)}}$.
	It is \emph{perfectly clustering} (for its alphabet) when $\pi \in S_\A$ is symmetric.
	The notions of clustering and perfectly clustering can be extended to multisets of words.
	
	\begin{example}
		Let us consider the three alphabets
		$\A = \{ {\tt a} < {\tt b} < {\tt n} \}$,
		$\A' = \{ {\tt a} < {\tt n} < {\tt b} \}$, and
		$\A'' = \{ {\tt n} < {\tt a} < {\tt b} \}$.
		The word $w = {\tt banana}$ is defined over each of the three alphabets.
		One has
		$\bwt{\A}{w} = {\tt nnbaaa}$,
		$\bwt{\A'}{w} = {\tt bnnaaa}$ and
		$\bwt{\A''}{w} = {\tt aabnna}$.
		So $w$ is perfectly clustering for $\A$ and $\A'$, but not clustering for $\A''$.
	\end{example}
	
	Over a binary alphabet (perfectly) clustering words coincide with powers of Christoffel words and their conjugates (\cite{MantaciRestivoRosoneSciortino07}).
	A characterization over larger alphabets in terms of factorization into palindroms is
	given in~\cite{LapointeReutenauer24} (see also~\cite{SimpsonPuglisi08}).
	
	The following result easily follows from Propositions~\ref{pro:bwt-conjugates} and~\ref{pro:bwt-primitive}.
	
	\begin{proposition}
		\label{pro:clusteringprimitive}
		Let $w = u^p \in \A^*$ with $u$ primitive.
		Then $w$ is $\pi$-clustering for $\A$ if and only if $u$ is $\pi$-clustering for $\A$.
	\end{proposition}

	\section{Interval Exchanges}
	\label{sec:ie}
	
	By an interval, we mean a left-closed and right-open interval over the real line.
	Let $\A$ be an ordered alphabet of cardinality $d$ and $\pi$ a permutation over $\A$.
	An ordered partition $(I_a)_{a \in \A}$ of an interval $I$ is such that $I_a$ is to the left of $I_b$ when $a < b$.
	The $d$-\emph{interval exchange transformation} (or $d$-\emph{IET} or just \emph{IET} in short) $T$ associated with a partition $(I_a)_{a \in \A}$ and a permutation $\pi$ is the piecewise translation on $I = [\ell, r)$ defined by
	$T(x) = x + \tau_a$ if $x \in I_a$, where
	$\tau_a = \sum_{\pi^{-1}(b) < \pi^{-1}(a)} |I_b| - \sum_{b < a} |I_b|$.
	Let $D(T) = \{ \sum_{b<a} |I_b| \; | \; a \in \A \} \setminus \{ \ell \}$ denote the set of \emph{formal discontinuities} of $T$.
	
	The \emph{orbit} of a point $x \in I$ under $T$ is the set $\{T^k(x) \; | \; k \in \Z \}$.
	The IET is \emph{periodic} if the orbit of any point $x \in I$ is finite.
	It is \emph{minimal} if the orbit of any point is dense in $I$.
	In this case, the permutation $\pi$ is irreducible (see, e.g.,~\cite{branching}).
	An IET is \emph{regular} (or satisfies \emph{Keane condition} or \emph{i.d.o.c.}) if the orbits of the formal discontinuities are infinite and disjoint.
	A regular IET is minimal and aperiodic~\cite{Keane75}, while the inverse is not true (see, e.g.,~\cite{bifixcodesiets}).
	
	A \emph{connection} of an IET $T$ is a triple $(x,y,n)$ where $x\in D(T^{-1})$, $y \in D(T)$, $n \ge 0$ and $T^n(x) = y$.
	When $n=0$, we call $x=y$ a $0$-\emph{connection}.
	A regular IET has no connection.
	
	Given an IET $T$ on $I$, for each point $x \in I$, we can assign to it an infinite word $\Omega_T(x) =w_0 w_1 w_2 \cdots$ describing its orbit, setting $w_k = a$ if $T^k(x) \in I_a$.
	This word is called the \emph{trajectory} of $x$ under $T$.
	The \emph{language} of an IET $T$ is
	$\LL(T) = \bigcup_{x \in I} \LL(\Omega_T(x))$.
	When $T$ is minimal or has only one periodic component (i.e., we have only one possible trajectory up to a shift), $\LL(T)$ does not depend on the choice of $x$.
	Moreover, in this case $\LL(T)$ is uniformly recurrent (see, e.g.,~\cite{branching}) and thus recurrent.
	
	Given an IET $T$ and a word $w =w_0 w_1 \cdots w_{n-1} \in \LL(T)$, we define the interval
	$I_w = I_{w_0} \cap T^{-1}(I_{w_1}) \cap \dots \cap T^{-(n-1)}(I_{w_{n-1}})$.
	By convention $I_\varepsilon = [\ell, r)$.
	For every point $x \in I_w$, the trajectory $\Omega_T(x)$ has $w$ as a prefix.

	\section{Discrete Interval Exchanges}
	
	A \emph{discrete interval exchange} (or \emph{DIET} in short) associated with the composition $(n_1, n_2, \ldots, n_d)$ of $n = \sum_{i=1}^d n_i$ and permutation $\pi \in S_d$ is the map $T(k) = k + t_i$ if $\sum_{j < i} n_j < k \le \sum_{j \le i} n_j$, where $t_i = \sum_{\pi^{-1}(j) < \pi^{-1}(i)} n_j - \sum_{j < i} n_j$.
	A DIET corresponds to an IET associated with a partition $(I_{a})_{a \in \A}$ and $\pi$, where $\A = \{ a_1 < \ldots < a_d \}$ and $|I_{a_i}| = n_i$.
	Note that each component of this corresponding IET is periodic;
	thus, in particular, a DIET is never minimal (nor regular).
	
	There is a strong link between clustering multisets of primitive words and DIETs.
	In fact, if a multiset $W \subset \A^*$ is $\pi$-clustering, then its Parikh vector gives a composition of $n = \sum_{w \in W} |w|$ that, along with $\pi$, defines a DIET.
	Similarly to IETs, we can encode the (periodic) trajectories by encoding each integer
	$k \in \left[ \sum_{j<i} n_j, \, \sum_{j \le i} n_j \right]$
	by the
	$i^{\rm th}$
	letter of the alphabet.
	
	In a symmetric way, it is possible to show that every DIET corresponds to a unique multiset of Lyndon words, with each orbit associated to a (clustering) Lyndon word.

	\begin{example}
		\label{ex:ebwt-diet}
		Let $W$ be the multiset of Example~\ref{ex:ebwt}.
		We can define a DIET $T$ associated with the composition $(4,2,1)$ of $7$ and the permutation $\pi = ({\tt c}, {\tt b}, {\tt a})$.
		The action of the DIET over $\{ 1, 2, \ldots, 7 \}$ is given by $\mu = (1,4,7)(2,5)(3,6) \in S_7$ (see left of Figure~\ref{fig:7diet}).
		Each orbit corresponds to one of the primitive words in $W$.
		For instance, the trajectory of $4$ is given by $\Omega(4) = ({\tt aca)}^\omega$, the infinite repetition of a conjugate of ${\tt aac}$.
		One can check that
		$I_{\tt a} = \{ 1,2,3,4 \},
		I_{\tt ab} = \{ 2,3 \},
		I_{\tt aac} = \{ 1 \}
		$.
		The corresponding IET is shown on the right of Figure~\ref{fig:7diet}.
	\end{example}
	
	\begin{figure}[ht]
		\centering
		\begin{tikzpicture}[scale=0.5]
			\node (u1) {$1$};
			\node (u2) [right=0cm of u1] {$2$};
			\node (u3) [right=0cm of u2] {$3$};
			\node (u4) [right=0cm of u3] {$4$};
			\node (u5) [right=0cm of u4] {$5$};
			\node (u6) [right=0cm of u5] {$6$};
			\node (u7) [right=0cm of u6] {$7$};
			\node[draw,rounded corners,blue,minimum width=1.4cm, minimum height=0.4cm, right=-0.3cm of u1] (ua) {};
			\node[draw,rounded corners,red,minimum width=0.7cm, minimum height=0.4cm, right=0.1cm of ua] (ub) {};
			\node[draw,rounded corners,green,minimum width=0.3cm, minimum height=0.4cm, right=0.1cm of ub] (uc) {};
			
			\node (b1) [below=0.1cm of u1] {$1$};
			\node (b2) [right=0cm of b1] {$2$};
			\node (b3) [right=0cm of b2] {$3$};
			\node (b4) [right=0cm of b3] {$4$};
			\node (b5) [right=0cm of b4] {$5$};
			\node (b6) [right=0cm of b5] {$6$};
			\node (b7) [right=0cm of b6] {$7$};
			\node[draw,rounded corners,green,minimum width=0.2cm, minimum height=0.4cm, right=-0.3cm of b1] (bc) {};
			\node[draw,rounded corners,red,minimum width=0.7cm, minimum height=0.4cm, right=0.1cm of bc] (bb) {};
			\node[draw,rounded corners,blue,minimum width=1.5cm, minimum height=0.4cm, right=0.1cm of bb] (ba) {};

			\draw[->, thick] (-.4, 0) .. controls (-.9, -0.1) and (-.9, -0.9) .. (-.4, -1) node[midway, left] {$T$};
		\end{tikzpicture}
		\qquad
		\begin{tikzpicture}[scale=0.5][x=5cm, y=1cm]
			\draw[thick,blue] (0,0.5) -- (4,0.5) node[midway,above] {${\tt a}$};
			\draw[thick,red] (4,0.5) -- (6,0.5) node[midway,above] {${\tt b}$};
			\draw[thick,green] (6,0.5) -- (7,0.5)  node[midway,above] {${\tt c}$};
			
			\draw[thick,green] (0,-0.7) -- (1,-0.7) node[midway,above] {${\tt c}$};
			\draw[thick,red] (1,-0.7) -- (3,-0.7) node[midway, above] {${\tt b}$};
			\draw[thick,blue] (3,-0.7) -- (7,-0.7)  node[midway,above] {${\tt c}$};
			
			\draw[->, thick] (-.3, 0.5) .. controls (-.9, 0.4) and (-.9, -0.6) .. (-.3, -0.7) node[midway, left] {$T$};
		\end{tikzpicture}
		\caption{A DIET (on the left) and its associated IET (on the right).}
		\label{fig:7diet}
	\end{figure}
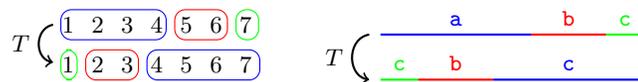
	
	In particular, one can view every primitive $\pi$-clustering word $w \in \A^*$ as a DIET associated with the composition $\Psi_{\A}(w)$ of $|w|$ and the permutation $\pi$, the permutation $\mu$ describing the action of such a DIET being circular (see~\cite{FerencziZamboni13} for a characterization of $\pi$-clustering words in terms of trajectories in IETs or DIETs).

	\section{Rauzy Induction}
	\label{sec:Rauzy}
	
	Let $\A = \{ a_1 < \ldots < a_d \}$ and $\pi$ a permutation of $\A$.
	Let $T$ be an IET over $[\ell, r)$ associated with $(I_a)_{a \in \A}$ and $\pi$.
	The \emph{transformation induced} by $T$ on a subinterval $J \subset I$ is the map $T': J \to J$ defined by $T'(z) = T^{\nu(z)} (z)$, where $\nu(z) = \min \{ n>0 \; | \; T^n(z) \in J \}$ is the \emph{first return map} of $T$ to $J$.
	Note that $\nu(z)$ is well-defined because IETs do not have wandering intervals (see, e.g.,~\cite{FerencziHubertZamboni24}).
	
	Rauzy induction is a procedure that associates to a regular IET $T$ associated with an alphabet $\A$, a sequence of regular IETs associated with the same reordered alphabet.
	
	The \emph{right Rauzy step}, is the mapping $\rho$ sending $T$ to the induced transformation $T'$ on $[\ell, r')$, where $r'$ is the rightmost between the points in $D(T) \cup D(T^{-1})$.
	Since $T$ is regular, it has no $0$-connections and $|I_{a_d}| \ne |I_{a_{\pi(d)}}|$.
	Moreover, since $\pi$ is irreducible we have $\pi(a_d) \ne a_d$.
	
	If $|I_{a_d}| > |I_{a_{\pi(d)}}|$ then $\Omega(x)$ starts with $a_{\pi(d)} a_d$ for every $x \in I_{a_{\pi(d)}}$; if $|I_{a_d}| < |I_{a_{\pi(d)}}|$ then $\Omega(x)$ starts with $a_{\pi(d)} a_d$ for every $x \in T^{-1}(I_{a_d})$.
	We can actually give a more precise description of the obtained induced intervals: in one case the order given by the alphabet stays the same while the one given by the permutation change, while in the other case the opposite happens.
	Recall that we identify $S_\A$ with $S_d$ when no confusion arises.
	
	\begin{lemma}
		\label{lem:rho}
		Let $T$ be a regular IET associated with $\A = \{ a_1 < \ldots < a_d \}$ and $\pi \in S_\A$.
		Let $h = \pi^{-1}(d)$.
		If $|I_{a_d}| > |I_{a_{\pi(d)}}|$, then $\rho(T)$ is the regular IET associated with $\A$ and $\pi' \in S_\A$ defined as
		$$
		\pi'(i) =
		\left\{
		\begin{array}{ll}
			\pi(i) & \quad {\rm if  } \quad i \le h \\
			\pi(d) & \quad {\rm if } \quad i = h+1 \\
			\pi(i)+1 & \quad {\rm if } \quad i > h+1
		\end{array}
		\right..
		$$
		
		If $|I_{a_d}| < |I_{a_{\pi(d)}}|$, then $\rho(T)$ is the regular IET associated with
		$$
		\A' = \{ a_1 < \ldots < a_h < a_d < a_{h+1} < \ldots < a_{d-1} \}
		$$
		and $\pi \in \A'$ defined as $\pi'(i) = \pi(i)$ for every $i$.
	\end{lemma}
	\begin{proof}
		Let $S = \rho(T)$.
		If $I_{a_d}$ is longer than $I_{a_{\pi(d)}}$, then the domain of $S$ is partitioned by $(I'_a)_{a \in \A}$, where all $I'_a = I_a$ but $I'_{a_d}$, which is cut of its final part.
		The first return map of $T$ into the domain of $S$ is given by $T^2(z)$ if $z \in I_{a_{\pi(d)}}$ and $T(z)$ elsewhere.
		Thus, $T(I_{a_d})$ is split into $S(I_{a_d})$ and $S(I_{a_{\pi(d)}})$.
		
		If $I_{a_d}$ is shorter than $I_{a_{\pi(d)}}$, then $S$ is defined as $T^2(z)$ if $z \in T^{-1}(I_{a_d})$ and $T(z)$ elsewhere.
		Thus, the subinterval $I_{a_{pi(d)}}$ for $T$ is split into $I'_{a_{\pi(d)}}$ and $I'_{a_d}$ in the partition associated with $S$.
		The interval $S(I'_{\pi(d)})$ will remain the rightmost (even though is smaller then $T(I_{\pi(d)})$), so the permutation does not change.
	\end{proof}
	
	The \emph{left Rauzy step} $\lambda$ is defined in a symmetric way considering the interval $[\ell', r)$, where $\ell' \ne \ell$ is the leftmost between the points in $D(T) \cup D(T^{-1})$ and the intervals considered are $I_{a_1}$ and $I_{\pi(a_1)}$.
	
	A symmetrical version of Lemma~\ref{lem:rho} holds.
	
	\begin{lemma}
		\label{lem:lambda}
		Let $T$ be a regular IET associated with $\A = \{ a_1 < \ldots < a_d \}$ and $\pi \in S_\A$.
		Let $h = \pi^{-1}(1)$.
		If $|I_{a_1}| > |I_{a_{\pi(1)}}|$, then $\lambda(T)$ is the regular IET associated with $\A$ and $\pi' \in S_\A$ defined as
		$$
		\pi'(i) =
		\left\{
		\begin{array}{ll}
			\pi(i)-1 & \quad {\rm if } \quad i < h-1 \\
			\pi(1) & \quad {\rm if } \quad i = h-1 \\
			\pi(i) & \quad {\rm if } \quad i \le h
		\end{array}
		\right. .
		$$
		
		If $|I_{a_1}| < |I_{a_{\pi(1)}}|$, then $\lambda(T)$ is the regular IET associated with
		$$
		\A' = \{ a_2 < \ldots < a_{h-1} < a_1 < a_h < \ldots < a_d \}
		$$
		and $\pi \in \A'$ defined as $\pi'(i) = \pi(i)$ for every $i$.
	\end{lemma}
	
	In~\cite{branching} it is shown that if $T$ is a regular IET, then for every $w \in \LL(T)$ the transformation induced by $T$ on $I_w$ is of the form $\chi(T)$, with $\chi \in \{ \rho, \lambda \}^*$, where each morphism corresponds to a Rauzy step.
	
	The following result is a consequence of Propositions 3.15, 4.12, 4.14 and Theorem 4.15 in~\cite{branching}.\footnote{The result being separated in several statements in~\cite{branching}, the authors managed to avoid using multiple indexes as it is done here.}
	
	\begin{proposition}[\cite{branching}]
		\label{pro:Iw}
		Let $T$ be a regular IET and $w \in \LL(T)$.
		The transformation induced by $T$ on $I_w$ is of the form $\chi(T)$, where $\chi \in \{ \rho, \lambda \}^*$.
		Moreover, let $\chi = \chi_n \circ \cdots \circ \chi_1$.
		Then the morphism $\theta = \theta_1 \circ \cdots \circ \theta_n$ is an automorphism of the free group sending $\A$ to $\R(w)$, where
		$$
		\theta_i =
		\left\{
		\begin{array}{ll}
			\alpha_{\pi(a^{(i)}_d),a^{(i)}_d} & \quad {\rm if } \quad \chi_i = \rho \quad {\rm  and } \quad |I_{a^{(i)}_d}| > |I_{\pi(a^{(i)}_d)}| \\
			\tilde{\alpha}_{a^{(i)}_d,\pi(a^{(i)}_d)} & \quad {\rm if } \quad \chi_i = \rho \quad {\rm and } \quad |I_{a^{(i)}_d}| < |I_{\pi(a^{(i)}_d)}| \\
			\alpha_{\pi(a^{(i)}_1),a^{(i)}_1} & \quad {\rm if } \quad \chi_i = \lambda \quad {\rm and } \quad |I_{a^{(i)}_1}| > |I_{\pi(a^{(i)}_1)}| \\
			\tilde{\alpha}_{a^{(i)}_1,\pi(a^{(i)}_1)} & \quad {\rm if } \quad \chi_i = \lambda \quad {\rm and } \quad |I_{a^{(i)}_1}| < |I_{\pi(a^{(i)}_1)}|
		\end{array}
		\right.
		$$
		and $\{ a^{(i)}_1 < \cdots < a^{(i)}_d \}$ is the alphabet associated to $\chi_i \circ \cdots \circ \chi_1(T)$.
	\end{proposition}

	\section{Rauzy Steps and Morphisms}
	
	In order to prove Theorem~\ref{thm:main}, we use the morphisms defined in Section~\ref{sec:preliminaries}\ to step back from $I_w$ to $[\ell, r)$.
	
	Let us show that under certain additional conditions, clustering is preserved by these morphisms.
	
	\begin{lemma}
		\label{lem:clustering}
		Let $w$ be a primitive word over $\A = \{a_1 < \ldots < a_d \}$.
		Suppose $w$ is $\pi$-clustering on $\A$ for some permutation $\pi$.
		\begin{enumerate}
			\item Let $\mu \in S_\A$ be a permutation.
			Then $\mu(w) \in \A'^*$ is $\pi'$-clustering, with $\A' = \{ \mu(a_1) < \ldots < \mu(a_d) \}$ and $\pi' \in S_{\A'}$.
			
			\item If $b = a_1$, $\pi^{-1}(a) = a_i$ and $\pi^{-1}(b) = a_{i+1}$ with $1 \le i < d$, then $\alpha_{a,b}(w) \in \A^*$ is $\pi'$-clustering, for a certain $\pi' \in S_\A$.
			
			\item If $b = a_d$, $\pi^{-1}(b) = a_i$ and $\pi^{-1}(a) = a_{i+1}$ with $1 \le i < d$, then $\alpha_{a,b}(w) \in \A^*$ is $\pi'$-clustering, for a certain $\pi' \in S_\A$.
			
			\item If $\pi^{-1}(b) = a_1$, $a=a_i$ and $b=a_{i+1}
			$ with $1 \le i < d$, then $\tilde{\alpha}_{a,b}(w) \in \A'^*$ is $\pi'$-clustering, for a certain $\pi' \in S_{\A'}$, where $\A' = \{a_i < a_1 < \ldots < a_{i-1} < a_{i+1} < \ldots < a_d \}$.
			
			\item If $\pi^{-1}(b) = a_d$, $b=a_i$ and $a=a_{i+1}$ with $1 \le i < d$, then $\tilde{\alpha}_{a,b}(w) \in \A'^*$ is $\pi'$-clustering, for a certain $\pi' \in S_{\A'}$, where $\A' = \{a_1 < \ldots < a_i < a_{i+2} < \ldots < a_d < a_{i+1} \}$.
		\end{enumerate}
	\end{lemma}
	
	\begin{proof}
		We proceed by addressing each case.
		\begin{enumerate}
			\item Since $w$ is $\pi$-clustering, $\bwt{\A}{w}$ consists of contiguous blocks - possibly of length $0$, if $w$ is not pangrammatic - of each letter of $\A$.
			Since an application of $\mu$ to $w$ simply amounts to letter renaming, we immediately obtain clustering of $\mu (w)$. Defining $\pi' = \mu \circ \pi \circ \mu^{-1} \in S_{\A'}$, via elementary permutation properties we see that $w$ is $\pi'$-clustering.
			
			\item Define $\pi' \in S_{\A}$ by altering $\pi$ such that $\pi^{-1} (a)$ and $\pi^{-1} (b)$ are adjacent in the cycle. If $\pi^{-1} (a) = a_i$ and $\pi^{-1} (b) = a_{i+1}$, then the blocks $\pi (a)$ and $\pi (b)$ appear consecutively in $\bwt{\A}{w}$.
			Replacing each $a$ by $ab$ in $w$ connects the blocks $\pi (a)$ and $\pi (b)$ into a single-block adjacency in $\bwt{\A}{\alpha_{a,b} (w)}$. Thus, $\alpha_{a,b} (w)$ is $\pi'$ clustering.
			
			\item This proof is identical to that of the second argument.
			
			\item We have $\pi(a_1) = b$. We define $\A' = \{a_i < a_1 < \ldots < a_{i-1} < a_{i+1} < \dots <a_d \}$ so that $a$ is the new smallest letter in $\A'$ with $b$ appearing later.
			Now let $\pi' \in S_{\A'}$ be the permutation such that $a$ is treated as the first letter and $b$ follows it somewhere in the cycle.
			Applying $\tilde{\alpha}_{a,b}$ to $w$ replaces every $a$ in $w$ with $ba$.
			In the Burrows-Wheeler transform, this forces the blocks of $\pi (a)$ and $\pi (b)$ to be contiguous.
			On the new order $\A'$ and under permutation $\pi'$, we immediately obtain the clustering via elementary Burrows-Wheeler transform properties.
			
			\item Let $\A' = \{ a_1 < \dots < a_i < a_{i+2} < \ldots < a_d < a_{i+1} \}$ so that $a$ is the largest letter of $\A'$.
			We define $\pi' \in S_{\A'}$ so that $\pi' (a_d) = b$ and $\pi' (b)$ such that $\pi' (b)$ sits next to $\pi' (a)$.
			Replacing $a$ by $ba$ again merges the blocks of $\pi(a)$ and $\pi(b)$ contiguously in the Burrows-Wheeler transform of $w$ over $\A'$.
			Thus, $\tilde{\alpha}_{a,b} (w)$ is $\pi'$-clustering.
		\end{enumerate}
	\end{proof}

	We are now able to prove our main result.
	
	\begin{proof}[of Theorem~\ref{thm:main}]
		Let $T$ be a regular IET and $w \in \LL(T)$.
		By Proposition~\ref{pro:Iw} there exist $\chi_1, \ldots, \chi_n \in \{ \rho, \lambda \}$ such that $\chi_n \circ \cdots \circ \chi_1(T)$ is the IET induced by $T$ on $I_w$.
		From the same proposition we obtain a morphism $\theta_i$ for each $1 \le i \le n$.
		
		By Lemmata~\ref{lem:rho} and~\ref{lem:clustering}, $\theta_i$ sends a clustering word in $\LL(\chi_{i-1} \circ \cdots \circ \chi_1(T))$ to a clustering word in $\LL(\chi_i \circ \cdots \circ \chi_1(T))$.
		Thus, by induction on $n$, $\theta(u)$ is clustering for every clustering word $u \in \LL(T)$.
		
		Since every letter is trivially clustering and $\R(w) = \{ \theta(a) \, | \, a \in \A \}$, we can conclude.
	\end{proof}
	
	\begin{example}
		\label{ex:krokzakrokem}
		Let $T$ be the IET associated with the alphabet $\{ {\tt a} < {\tt b} < {\tt c} \}$ and the permutation $\pi = ({\tt b,c,a})$, with $|I_{\tt a}| = 1-2\alpha$, $|I_{\tt b}| = |I_{\tt c}| = \alpha$, where $\alpha = \frac{3-\sqrt{5}}{2}$.
		The IET is regular, since it is just the rotation of the irrational angle $\alpha$.
		The transformation induced to the subinterval $I_{\tt b}$ is $\chi (T)$, with $\chi = \lambda^2 \circ \rho^2$.
		We have $\theta(\{ {\tt a,b,c} \}) = \R({\tt b}) = \{ {\tt bac, b, bacc} \}$, where
		$\theta =
		\alpha_{{\tt a},{\tt c}} \circ
		\tilde{\alpha}_{{\tt c},{\tt a}} \circ
		\tilde{\alpha}_{{\tt a},{\tt b}} \circ
		\tilde{\alpha}_{{\tt c},{\tt b}}$.
	\end{example}
	
	\begin{figure}[ht]
		\centering
		\begin{tikzpicture}[x=8cm, y=1cm]
			\draw[thick,red] (0,0) -- (1-2*0.382,0) node[midway,above] {${\tt a}$};
			\draw[thick,blue] (1-2*0.382,0) -- (1-0.382,0) node[midway,above] {${\tt b}$};
			\draw[thick,green] (1-0.382,0) -- (1,0)  node[midway,above] {${\tt c}$};
			\draw[thick,blue] (0,-0.5) -- (0.382,-0.5) node[midway,above] {${\tt b}$};
			\draw[thick,green] (0.382,-0.5) -- (2*0.382,-0.5) node[midway, above] {${\tt c}$};
			\draw[thick,red] (2*0.382,-0.5) -- (1,-0.5)  node[midway,above] {${\tt a}$};
			\draw[->] (-.02,0-.1) to[bend right] node[midway, left] {} (-.02,-0.5+.1);
			
			\draw[thick,red] (0,-.5-.7) -- (1-2*0.382,-.5-.7) node[midway,above] {${\tt a}$};
			\draw[thick,blue] (1-2*0.382,-.5-.7) -- (1-0.382,-.5-.7) node[midway,above] {${\tt b}$};
			\draw[thick,green] (1-0.382,-.5-.7) -- (2*0.382,-.5-.7)  node[midway,above] {${\tt c}$};
			\draw[thick,blue] (0,-.5*2-.7) -- (0.382,-.5*2-.7) node[midway,above] {${\tt b}$};
			\draw[thick,green] (0.382,-.5*2-.7) -- (4*0.382-1,-.5*2-.7) node[midway, above] {${\tt c}$};
			\draw[thick,red] (4*0.382-1,-.5*2-.7) -- (2*0.382,-.5*2-.7)  node[midway,above] {${\tt a}$};
			\draw[->] (-.02, -.5-.7-.1) to[bend right] node[midway, left] {} (-.02,-0.5*2-.7+.1);
			
			\draw[thick,red] (0,-.5*2-.7*2) -- (2-5*0.382,-.5*2-.7*2) node[midway,above] {${\tt a}$};
			\draw[thick,green] (2-5*0.382,-.5*2-.7*2) -- (1-2*0.382,-.5*2-.7*2) node[midway,above] {${\tt c}$};
			\draw[thick,blue] (1-2*0.382,-.5*2-.7*2) -- (1-0.382,-.5*2-.7*2)  node[midway,above] {${\tt b}$};
			\draw[thick,blue] (0,-.5*3-.7*2) -- (0.382,-.5*3-.7*2) node[midway,above] {${\tt b}$};
			\draw[thick,green] (0.382,-.5*3-.7*2) -- (4*0.382-1,-.5*3-.7*2) node[midway, above] {${\tt c}$};
			\draw[thick,red] (4*0.382-1,-.5*3-.7*2) -- (1-0.382,-.5*3-.7*2)  node[midway,above] {${\tt a}$};
			\draw[->] (-.02, -.5*2-.7*2-.1) to[bend right] node[midway, left] {} (-.02,-0.5*3-.7*2+.1);
			
			\draw[thick,green] (2-5*0.382,-.5*3-.7*3) -- (1-2*0.382,-.5*3-.7*3) node[midway,above] {${\tt c}$};
			\draw[thick,red] (1-2*0.382,-.5*3-.7*3) -- (3-7*0.382,-.5*3-.7*3)  node[midway,above] {${\tt a}$};
			\draw[thick,blue] (3-7*0.382,-.5*3-.7*3) -- (1-0.382,-.5*3-.7*3) node[midway,above] {${\tt b}$};
			\draw[thick,blue] (2-5*0.382,-.5*4-.7*3) -- (0.382,-.5*4-.7*3) node[midway,above] {${\tt b}$};
			\draw[thick,green] (0.382,-.5*4-.7*3) -- (4*0.382-1,-.5*4-.7*3) node[midway, above] {${\tt c}$};
			\draw[thick,red] (4*0.382-1,-.5*4-.7*3) -- (1-0.382,-.5*4-.7*3)  node[midway,above] {${\tt a}$};
			\draw[->] (2-5*0.382-.02, -.5*3-.7*3-.1) to[bend right] node[midway, left] {} (2-5*0.382-.02,-0.5*4-.7*3+.1);
			
			\draw[thick,red] (1-2*0.382,-.5*4-.7*4) -- (3-7*0.382,-.5*4-.7*4)  node[midway,above] {${\tt a}$};
			\draw[thick,green] (3-7*0.382,-.5*4-.7*4) -- (2-4*0.382,-.5*4-.7*4) node[midway,above] {${\tt c}$};
			\draw[thick,blue] (2-4*0.382,-.5*4-.7*4) -- (1-0.382,-.5*4-.7*4) node[midway,above] {${\tt b}$};
			\draw[thick,blue] (1-2*0.382,-.5*5-.7*4) -- (0.382,-.5*5-.7*4) node[midway,above] {${\tt b}$};
			\draw[thick,green] (0.382,-.5*5-.7*4) -- (4*0.382-1,-.5*5-.7*4) node[midway, above] {${\tt c}$};
			\draw[thick,red] (4*0.382-1,-.5*5-.7*4) -- (1-0.382,-.5*5-.7*4)  node[midway,above] {${\tt a}$};
			\draw[->] (1-2*0.382-.02, -.5*4-.7*4-.1) to[bend right] node[midway, left] {} (1-2*0.382-.02,-0.5*5-.7*4+.1);
			
			\draw[->] (-.1, -0.4) to[bend right] node[midway, left] {$\rho$} (-.1,-.1-0.5-0.7);
			\draw[->] (-.1, -0.4-.7-.5) to[bend right] node[midway, left] {$\rho$} (-.1,-.1-0.5-0.7*2-.5);
			\draw[->] (-.1, -0.4-.7*2-.5*2) to[bend right] node[midway, left] {$\lambda$} (-.1,-.1-0.5-0.7*3-.5*2);
			\draw[->] (-.1, -0.4-.7*3-.5*3) to[bend right] node[midway, left] {$\lambda$} (-.1,-.1-0.5-0.7*4-.5*3);
			
			\draw[<-] (1+.1, -0.4) to[bend left] node[midway, right] {$\alpha_{{\tt a},{\tt c}}$} (1+.1,-.1-0.5-0.7);
			\draw[<-] (1+.1, -0.4-.7-.5) to[bend left] node[midway, right] {$\tilde{\alpha}_{{\tt c},{\tt a}}$} (1+.1,-.1-0.5-0.7*2-.5);
			\draw[<-] (1+.1, -0.4-.7*2-.5*2) to[bend left] node[midway, right] {$\tilde{\alpha}_{{\tt a},{\tt b}}$} (1+.1,-.1-0.5-0.7*3-.5*2);
			\draw[<-] (1+.1, -0.4-.7*3-.5*3) to[bend left] node[midway, right] {$\tilde{\alpha}_{{\tt c},{\tt b}}$} (1+.1,-.1-0.5-0.7*4-.5*3);

		\end{tikzpicture}
		\caption{Series of Rauzy steps and their associated morphisms.}
		\label{fig:krokzakrokem}
	\end{figure}
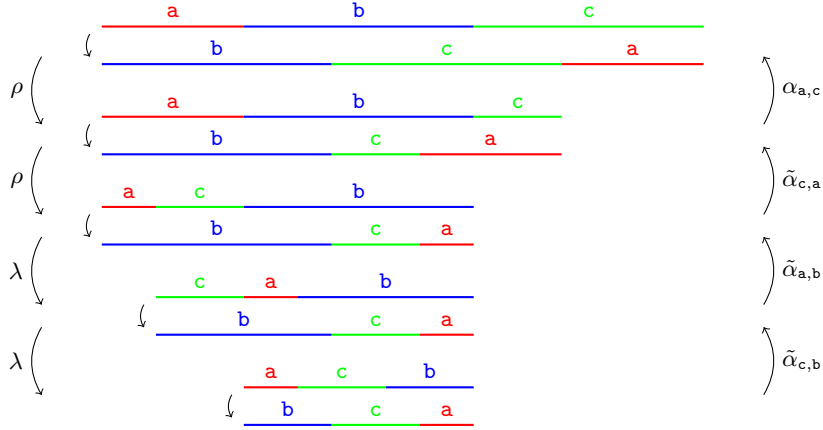

	\section{See the Forest for the IETs}
	\label{sec:forest}
	
	To conclude our study of clustering words in languages generated by interval exchanges, we now slightly vary our trajectory to discuss the connection between IETs (resp. DIETs) and dendricity.
	
	Let $\LL \subset \A^*$ be a language.
	The \emph{extension graph} $\G(w)$ of a word $w \in \LL$ is the undirected bipartite graph having as vertices the disjoint union of $L(w) = \{a \in \A \;| \; aw \in \LL \}$ and $R(w) = \{b \in \A \; | \; wb \in \LL \}$, and edges $B(w) = \{ (a,b) \in \A^2 \; | \; awb \in \LL \}$.
	The graph $\G(w)$ is \emph{compatible} with two orders $<_1$ and $<_2$ on $\A$ if for every $(a,b), (c,d) \in B(w)$ one has the implication $a <_1 c \Longrightarrow b \le_2 d$.
	
	A language $\LL$ is said to be \emph{dendric} if the extension graph of every $w \in \LL$ is a tree, i.e., acyclic and connected (whence the original name \emph{tree set} in~\cite{acyclicconnectedtree}).
	Following the same hellenophilic spirit, we call a language \emph{alsinic} if the extension graph of every word in it is a forest, i.e., acyclic but not necesserily connected.
	A language $\LL$ is \emph{ordered dendric} (resp. \emph{ordered alsinic}) for two orders $<_1$ and $<_2$ if every $\G(w)$, with $w \in \LL$, is compatible for $<_1$ and $<_2$ (in~\cite{bifixcodesiets} the term \emph{planar tree} was used since the edges do not cross).
	Examples of dendric but not ordered dendric languages are given by Arnoux-Rauzy words on more than two letters~\cite{bifixcodesiets}.
	
	Ordered alsinic languages are strictly linked to IETs.
	
	\begin{theorem}[\cite{FerencziHubertZamboni24,FerencziZamboni08}]
		$\LL$ is the language of an IET $T$ if and only if it is a recurrent ordered alsinic language.
		
		$\LL$ is the language of a minimal IET if and only if it is an aperiodic, uniformly recurrent ordered alsinic language.
		
		$\LL$ is the language of a regular IET if and only if it is a uniformly recurrent ordered dendric set.
	\end{theorem}
	
	As seen above, clustering words are associated with DIETs, and these can be seen as IETs where intervals have integer lengths.
	The following result from~\cite{FerencziHubertZamboni23} make this link explicit.
	Let us denote by $<_{\A}$ the order on the alphabet $\A$ and by $<_\pi$ the order given by $a <_\pi b$ when $\pi^{-1}(a) <_\A \pi^{-1}(b)$.
	
	\begin{theorem}[\cite{FerencziHubertZamboni23}]
		\label{thm:FHZ23}
		A word $w \in \A^*$ is $\pi$-clustering if and only if for every bispecial word $v \in \LL(w^\omega)$, the graph $\G(v)$ is compatible with the orders $<_\pi$ and $<_\A$.
	\end{theorem}
	
	The following result easily follows.
	
	\begin{corollary}
		A word $w \in \A^*$ is $\pi$-clustering if and only if $\LL(w^\omega)$ is ordered alsinic for the orders $<_\pi$ and by $<_\A$.
	\end{corollary}
	\begin{proof}
		If $u \in \LL(w^\omega)$ is not left special (resp., not right special) then $\G(u)$ is a tree with only one vertex to the left (resp., to the right).
		It is thus possible to order the vertices to the right (resp. to the left) using $<_\A$ (resp. $<_\pi$).
	\end{proof}
	
	Following the same argument seen in Section~\ref{sec:ie}, Theorem~\ref{thm:FHZ23} can be generalized to multiset of words.
	
	\begin{example}
		\label{ex:extensiongraph}
		Let $W$ and $T$ as in Example~\ref{ex:ebwt-diet}.
		Then ${\tt c} <_\pi {\tt b} <_\pi {\tt a}$ and ${\tt a} <_{\A} {\tt b} <_{A} {\tt c}$.
		The extension graphs of the empty word and the letter ${\tt a}$ are shown in Figure~\ref{fig:extensiongraph}.
		It is easy to show that $\G(w)$ contains only one edge for every $w \in \LL \setminus \{ \varepsilon, {\tt a} \}$.
	\end{example}
	
	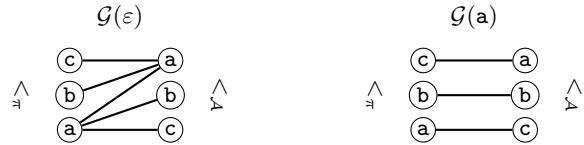
\begin{figure}[ht]
		\centering
		\tikzset{node/.style={draw, circle, inner sep=0.4mm}}
		\tikzset{title/.style={minimum size=0.5cm,inner sep=1pt}}
		\begin{tikzpicture}
			\node[title] (Ee) {$\G(\varepsilon)$};
			\node[node](ecl) [below left = 0.2cm and 0.2cm of Ee] {${\tt c}$};
			\node[node](ebl) [below = 0.1cm of ecl] {${\tt b}$};
			\node[node](eal) [below = 0.1cm of ebl] {${\tt a}$};
			\node[title](ePi) [left=0.2cm of ebl] {\eqmathbox{\rotatebox[origin=c]{-90}{$<_\pi$}}};
			\node[node](ear) [below right = 0.2cm and 0.2cm of Ee] {${\tt a}$};
			\node[node](ebr) [below = 0.1cm of ear] {${\tt b}$};
			\node[node](ecr) [below = 0.1cm of ebr] {${\tt c}$};
			\node[title](eA) [right=0.2cm of ebr] {\eqmathbox{\rotatebox[origin=c]{-90}{$<_\A$}}};
			\path
			(ecl) edge[thick] (ear)
			(ebl) edge[thick] (ear)
			(eal) edge[thick] (ear)
			(eal) edge[thick] (ebr)
			(eal) edge[thick] (ecr);
			
			\node[title] (Ea)[right = 4cm of Ee] {$\G({\tt a})$};
			\node[node](acl) [below left = 0.2cm and 0.2cm of Ea] {${\tt c}$};
			\node[node](abl) [below = 0.1cm of acl] {${\tt b}$};
			\node[node](aal) [below = 0.1cm of abl] {${\tt a}$};
			\node[title](aPi) [left=0.2cm of abl] {\eqmathbox{\rotatebox[origin=c]{-90}{$<_\pi$}}};
			\node[node](aar) [below right = 0.2cm and 0.2cm of Ea] {${\tt a}$};
			\node[node](abr) [below = 0.1cm of aar] {${\tt b}$};
			\node[node](acr) [below = 0.1cm of abr] {${\tt c}$};
			\node[title](aA) [right=0.2cm of abr] {\eqmathbox{\rotatebox[origin=c]{-90}{$<_\A$}}};
			\path
			(acl) edge[thick] (aar)
			(abl) edge[thick] (abr)
			(aal) edge[thick] (acr);
			
		\end{tikzpicture}
		\caption{Extension graphs of $\varepsilon$ and ${\tt a}$ in $\LL(T)$, with $T$ as in Example~\ref{ex:extensiongraph}.}
		\label{fig:extensiongraph}
	\end{figure}
	
	\begin{proposition}
		\label{pro:multiclustering}
		If a multiset $W \subset \A^*$ is $\pi$-clustering, then every $w \in W$ is $\pi_w$-clustering, with $\pi_w$ the restriction of $\pi$ to the letters appearing in $w$.
	\end{proposition}
	\begin{proof}
		The result easily follows by considering the DIET associated with $W$ and noticing that each word $w \in W$ corresponds to exactly one orbit of the DIET.
	\end{proof}
	
	Note that the opposite of Proposition~\ref{pro:multiclustering} is not true.
	
	\begin{example}
		\label{ex:multinonclustering}
		Let $\A = \{ {\tt a} < {\tt b} \}$.
		The words $w_1 = {\tt ab}, w_2 = {\tt aab}$ are $\pi$-clustering with $\pi = ({\tt ba})$.
		The  multiset $W = \{ w_1, w_2 \}$ is not clustering since
		$\ebwt{\A}{W} = {\tt babaa}$.
		Note also that $\LL(W^\omega) = \LL(w_1^\omega) \cup \LL(w_2^\omega)$ is not ordered dendric as one can easily check by considering $\G(\varepsilon)$, $\G({\tt a})$ and $\G({\tt aba})$ (see Figure~\ref{fig:multinonclustering}).
	\end{example}
	
	\begin{figure}[ht]
		\centering
		\tikzset{node/.style={draw, circle, inner sep=0.4mm}}
		\tikzset{title/.style={minimum size=0.5cm,inner sep=1pt}}
		\begin{tikzpicture}
			\node[title] (Ee) {$\G(\varepsilon)$};
			\node[node](ebl) [below left = 0.2cm and 0.2cm of Ee] {${\tt b}$};
			\node[node](eal) [below = 0.1cm of ebl] {${\tt a}$};
			\node[node](ear) [below right = 0.2cm and 0.2cm of Ee] {${\tt a}$};
			\node[node](ebr) [below = 0.1cm of ear] {${\tt b}$};
			\path
			(ebl) edge[thick] (ear)
			(eal) edge[thick] (ear)
			(eal) edge[thick] (ebr);
			
			\node[title] (Ea)[right = 3cm of Ee] {$\G({\tt a})$};
			\node[node](abl) [below left = 0.2cm and 0.2cm of Ea] {${\tt b}$};
			\node[node](aal) [below = 0.1cm of abl] {${\tt a}$};
			\node[node](aar) [below right = 0.2cm and 0.2cm of Ea] {${\tt a}$};
			\node[node](abr) [below = 0.1cm of aar] {${\tt b}$};
			\path
			(abl) edge[thick] (aar)
			(abl) edge[thick] (abr)
			(aal) edge[thick] (abr);
			
			\node[title] (Eaba)[right = 3cm of Ea] {$\G({\tt aba})$};
			\node[node](ababl) [below left = 0.2cm and 0.1cm of Eaba] {${\tt b}$};
			\node[node](abaal) [below = 0.1cm of ababl] {${\tt a}$};
			\node[node](abaar) [below right = 0.2cm and 0.1cm of Eaba] {${\tt a}$};
			\node[node](ababr) [below = 0.1cm of abaar] {${\tt b}$};
			\path
			(abaal) edge[thick] (abaar)
			(ababl) edge[thick] (ababr);
			
		\end{tikzpicture}
		\caption{Extension graphs of $\varepsilon$, ${\tt a}$ and ${\tt aba}$ as in Example~\ref{ex:multinonclustering}.}
		\label{fig:multinonclustering}
	\end{figure}
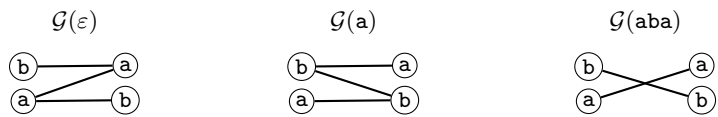

	\section{Concluding Remarks}
	\label{sec:conclusions}
	
	Our approach not only leverages the deep combinatorial structure inherent in Rauzy induction, but also sets the stage for potential generalizations to broader classes of interval exchange transformations in future research.
	
	The question of whether such return words are perfectly clustering when the interval exchange transformation is symmetric remains open. More generally, are return words in languages of IETs associated with a permutation $\pi$ necessarily $\pi$-clustering?
	It is reasonable to anticipate that the techniques developed here could aid in answering these questions in future contributions.

	\addcontentsline{toc}{section}{References}
	\bibliographystyle{plain}
	\bibliography{references.bib}
	
\end{document}